%% file: sketch.tex
\documentclass[11pt]{article}

\usepackage{geometry}
\usepackage{amssymb}
\usepackage{mathrsfs}
\usepackage{amsmath}
\usepackage[latin1]{inputenc}

\newtheorem{theorem}{Theorem}[section]
\newtheorem{corollary}[theorem]{Corollary}
\newtheorem{lemma}[theorem]{Lemma}
\newtheorem{proposition}[theorem]{Proposition}
\newtheorem{claim}[theorem]{Claim}
\newtheorem{definition}[theorem]{Definition}

\def\squarebox#1{\hbox to #1{\hfill\vbox to #1{\vfill}}}
\newcommand{\qed}{\hspace*{\fill}
\vbox{\hrule\hbox{\vrule\squarebox{.667em}\vrule}\hrule}\smallskip}

\newenvironment{proof}{\noindent{\bf Proof:~~}}{\(\qed\)}

\begin{document}
\title{Faster and Simpler Sketches of Valuation Functions}

\author{Keren Cohavi \and Shahar Dobzinski\thanks{Weizmann Institute of Science. Emails: \texttt{\{keren.cohavi, shahar.dobzinski\}@weizmann.ac.il}.}}

\maketitle

\begin{abstract}
We present fast algorithms for sketching valuation functions. Let $N$ ($|N|=n$) be some ground set and $v:2^N\rightarrow \mathbb R$ be a function. We say that $\tilde v:2^N\rightarrow \mathbb R$ is an \emph{$\alpha$-sketch} of $v$ if for every set $S$ we have that $\frac {v(S)} {\alpha} \leq \tilde v(S) \leq v(S)$ and $\tilde v$ can be described in $poly(n)$ bits.

Goemans et al. [SODA'09] showed that if $v$ is submodular then there exists an $\tilde O(\sqrt n)$-sketch that can be constructed using polynomially many value queries (this is essentially the best possible, as Balcan and Harvey [STOC'11] show that no submodular function admit an $n^{\frac 1 3 - \epsilon}$-sketch). Based on their work, Balcan et al. [COLT'12] and Badanidiyuru et al [SODA'12] show that if $v$ is subadditive then there exists an $\tilde O(\sqrt n)$-sketch that can be constructed using polynomially many demand queries. All previous sketches are based on complicated geometric constructions. The first step in their constructions is proving the existence of a good sketch by finding an ellipsoid that ``approximates'' $v$ well (this is done by applying John's theorem to ensure the existence of an ellipsoid that is ``close'' to the polymatroid that is associated with $v$). The second step is to show that this ellipsoid can be found efficiently, and this is done by repeatedly solving a certain convex program to obtain better approximations of John's ellipsoid.

In this paper, we give a significantly simpler, non-geometric proof for the existence of good sketches, and utilize the proof to obtain much faster algorithms that match the previously obtained approximation bounds. Specifically, we provide an algorithm that finds $\tilde O(\sqrt n)$-sketch of a submodular function with only $\tilde O(n^\frac{3}{2})$ value queries, and we provide an algorithm that finds $\tilde O(\sqrt n)$-sketch of a subadditive function with $O(n)$ demand and value queries.
\end{abstract}

\section{Introduction}
In this paper we tackle the problem of sketching valuation functions. Consider a function $v:2^N\rightarrow \mathbb R_+$ that is defined over a ground set of $N$, $|N|=n$, elements. Representing $v$ requires specifying $2^n$ numbers, which is too large for most problems. For example, if the function is submodular then we usually require that optimization algorithms will run in time $poly(n)$.

One way to overcome the representation obstacle is to assume that $v$ is given as a black box and develop algorithms that query $v$ no more than $poly(n)$ times. Goemans et al \cite{GHIM09} suggest the following alternative approach: instead of $v$ we can use $\tilde v$ -- a function in which for all $S$, $\tilde v(S)$ is close to $v(S)$ and $\tilde v$ can be succinctly described. More formally, $\tilde v:2^N\rightarrow \mathbb R_+$ is an $\alpha$-\emph{sketch}\footnote{The paper \cite{GHIM09} says that in this case $\tilde v$ \emph{approximates $v$ everywhere}, but in this paper we use the term \emph{sketch} that was suggested later in \cite{BDFK+12}.} of $v$ if $\tilde v$ can be represented in polynomial space and for every bundle $S$ we have that $\frac {v(S)} \alpha \leq \tilde v(S) \leq v(S)$.

Goemans et al \cite{GHIM09} show that not only that every submodular function $v$ admits an $\tilde O(\sqrt n)$-sketch, but also that if $v$ is given to us as a black box we can find this sketch by making only $poly(n)$ value queries (given $S$, what is $v(S)$?). Their construction is technical and involved. Very roughly speaking, existence of a good sketch is proved by considering a polymatroid associated with $v$, and finding an ellipsoid $E$ that approximates well the polymatroid using John's theorem. Algorithmically finding an ellipsoid that approximates $E$ is done by repeatedly solving several instances of a certain convex optimization problem, to get closer and closer to $E$. We refer the reader to \cite{GHIM09} for full details on their approach.

Balcan and Harvey \cite{BH11} prove that the construction above is essentially the best possible, in the sense that there are submodular functions (which are even rank functions of a certain family of matroids) for which no $n^{\frac 1 3-\epsilon}$-sketch exists. Thus, the focus of this research line has switched to finding good sketches for larger classes of valuations\footnote{A very related question is learning valuation functions -- see, e.g., \cite{BFIW11,BH11,FV13, FK14}.}. Balcan et al. \cite{BFIW11} and Badanidiyuru et al \cite{BDFK+12} show that $\tilde O(\sqrt n)$-sketches exist for the larger class of subadditive valuations. Furthermore, these sketches can be found using a polynomial number of demand queries (the power of value queries is weak in this setting -- it is known that with polynomially many value queries only trivial sketches of subadditive functions can be obtained \cite{BDFK+12}). These papers follow the basic ellipsoidal-based approach of \cite{GHIM09}, developing it even further to support subadditive functions.

Our contribution is two-fold. First, we present a radically different proof for all the above positive results. Our proof is much simpler, uses first principles only, and in particular avoids complicated geometric arguments. Moreover, although all of the above algorithmic sketches run formally in polynomial time, the hidden constants are very large. A concrete benefit of our proof is that it enables us to develop much faster sketching algorithms, hence joining the recent effort of obtaining faster versions of fundamental algorithms in submodular optimization (e.g, \cite{BV14, BBCL14}).

\begin{theorem} The following two statements hold:
\begin{itemize}
\item There exists an algorithm that finds $\tilde O(\sqrt n)$-sketch for any submodular function using $O\left(n^\frac{3}{2}\cdot \log^{3}\left(n\right)\right)$ value queries.
\item There exists an algorithm that finds $\tilde O(\sqrt n)$-sketch for any subadditive function using $O(n)$ value and demand queries.
\end{itemize}
\end{theorem}

In fact, we only provide and analyze a single algorithm that calls in a black box manner to subroutines such as maximizing a function subject to a cardinality constraint. We obtain both parts of the theorem by using known implementations of these black boxes. In addition, we use the same algorithm to obtain $\tilde O(\sqrt n)$-sketch for matroid rank functions using only $\tilde O(n)$ value queries.

\subsection*{An Overview of our Construction}

We now give a brief overview of our construction. Recall that a valuation $a$ is \emph{additive} if for every set $S$ we have that $a(S)=\Sigma_{j\in S}a(\{j\})$. Also, recall that a valuation function $v$ is XOS if there exist additive valuations $a_1,\ldots, a_t$ such that $v(S)=\max_ra_r(S)$. (It is known that every XOS function is also a subadditive function, and every submodular function is also an XOS function.) For simplicity of presentation, in the next few paragraphs we focus on the case where for each item $j$ and valuation $a_r$ we have that\footnote{This is less restrictive than it sounds -- not only that this class contains all matroid rank functions, but moreover one can apply standard arguments to show that the existence of an $\alpha$-sketch for this case implies an $\tilde O(\alpha)$-sketch for the general XOS valuations. However, these arguments do not yield an efficient algorithm.} $a_r(\{j\})\in \{0,1\}$.

Let us now discuss a procedure that given a value $k$ correctly sketches all sets $S$ such that $v(S)=k$. If $k\leq \sqrt n$, to approximate the value of a given bundle $S$ it is enough to decide whether there exists a single item $j\in S$ such that $v(\{j\})=1$ or whether for all $j\in S$ we have that $v(\{j\})=0$. In the former case, our sketch can return that the value of $S$ is $1$, and we maintain a $\sqrt n$-approximation. In the latter, we know that $v(S)$ is $0$. Notice that the information needed to decide between the cases can be written down using only $n$ bits, where the $j$'th bit specifies whether $v(\{j\})=1$ or not.

The heart of the proof is the trickier case where $k>\sqrt n$. In this case our sketch will include a maximal number of disjoint bundles $T_1,\ldots, T_l\subseteq N$ where for each bundle $T_i$ we have that $v(T_i)=\frac k 2$ and $|T_i|=\frac k 2$. Observe that since the $T_i$'s are disjoint and since $\frac k 2\geq \frac {\sqrt n} 2$, we have that $l\leq 2\sqrt n$. Consider some bundle $S$ where $v(S)=k$ and let $R\subseteq S$, $|R|=k$, be such that $v(R)=k$. Notice that such $R$ exists since the value of every item in any of the additive valuations is in $\{0,1\}$. Furthermore, this also implies that for each $R'\subseteq R$ we have that $v(R')=|R'|$.

Let $R'=R\cap (\cup_iT_i)$. First, observe that $|R'|\geq k/2$. If this were not the case, we have that $|R-R'|\geq k/2$. This is a contradiction to our assumption that $T_1,\ldots, T_l$ is a maximal disjoint set of bundles such that $v(T_i)=\frac k 2$, since we could have added any subset of size $\frac k 2$ from $R-R'$ to the $T_i$'s. Since there are $l\leq 2\sqrt n$ disjoint sets, there must be one set $T_i$ such that $|R'\cap T_i| \geq \frac {\frac k 2} {2\sqrt n}=\frac {k} {4\sqrt n}$. Thus, for every given $S$ our sketch will return the size of largest intersection between $S$ and the $T_i$'s, and this size is guaranteed to be within a factor of $O(\sqrt n)$ of $v(S)$.

Finally, we run the procedure above for all possible values of $1\leq k\leq n$. Given a bundle $S$, we return the maximum value of $S$ that was computed by any of the sketches constructed by the procedure. By the discussion above, this value is within a factor of $O(\sqrt n)$ of $v(S)$.

\input{prelim}

\input{construction}

\bibliographystyle{plain}
\bibliography{bib}

\appendix

\input{appendix}

\end{document}

%% file: prelim.tex
\section{Preliminaries}

In this paper we consider valuation functions $v:\,2^{N}\rightarrow\mathbb{R}_{+}$ where $N=\left\{ {1},\ldots,{n}\right\} $. We assume that $v$ is normalized ($v(\emptyset)=0$) and monotone (for every $S\subseteq T\subseteq N$ it holds that $v\left(S\right)\leq v\left(T\right)$). In addition, we consider the following restrictions on the valuation functions:
\begin{itemize}
\item \textbf{Subadditivity:} for every $S,T\subseteq N, v\left(S\right)+v\left(T\right)\geq v\left(S\cup T\right)$.
\item \textbf{$\beta$-XOS:} There exists a set of additive functions $a_{1},\ldots,a_{t}$ such that for every bundle $S$ the following two conditions hold:
\begin{itemize}

\item
Let $a_{l}\left(S\right)=\max_{r}a_{r}(S)$, then $a_{l}\left(S\right)\leq v\left(S\right)\leq \beta\cdot a_{l}\left(S\right)$.

\item
For every $T\subseteq S$, $v\left(T\right)\geq \underset{j\in T} \sum{a_l\left(j\right)}$.

\end{itemize}

Given a valuation function which is $\beta-XOS$ and a bundle $S$, an approximate XOS clause of $S$ will be an additive function for which both the conditions above hold. A function $v$ which is $\beta-XOS$ with $\beta=1$ is called simply an XOS function.
\item \textbf{Submodularity:} for every $S,T\subseteq N, v\left(S\right)+v\left(T\right)\geq v\left(S\cup T\right)+v\left(S\cap T\right)$. An equivalent and more intuitive definition of submodular functions
is the property of decreasing marginal values: for every $X\subseteq Y\subseteq N$,
$z\in N\backslash Y$ it holds that $v\left(Y\cup\left\{ z\right\} \right)-v\left(Y\right)\leq v\left(X\cup\left\{ z\right\} \right)-v\left(X\right)$.

\item \textbf{Matroid rank function:} Let $M=\left(\mathcal N, \mathcal I\right)$ be a matroid. The rank function of the matroid, $r:\,2^{\mathcal N}\rightarrow\mathbb{N}$, is defined for every $S\subseteq \mathcal N$ as $r\left(S\right)=\max_{T\subseteq S}\{|T|,\,T\in\mathcal I\}$.

\end{itemize}

Every XOS function is also a subadditive function, and every submodular function is also an XOS function \cite{LLN01}. Furthermore, every subadditive function is $O(\log n)$-XOS \cite{D07}. It is also well known that every matroid rank function is submodular.

We say that a function $\tilde v:2^N\rightarrow \mathbb R_+$ is an $\alpha$-\emph{sketch} of $v$ if $\tilde v$ can be represented in $poly(n)$ space and for every bundle $S$ we have that $\frac {v(S)} \alpha \leq \tilde v(S) \leq v(S)$. We assume that $v$ is represented by a black box that can answer a specific type of queries. One type is \emph{value query}: given $S$, return $v(S)$. Another type is \emph{demand query}: given prices per item $p_1,\ldots, p_j$, return a most profitable bundle $S\in \arg\max_T v(T)-\Sigma_{j\in T}p_j$.

Recall that a sketch $\tilde v$ has to be represented in $poly(n)$ space. However, if the ratio $\frac {\max_j{v(\{j\})}} {\min_j{v(\{j\})}}$ is large (e.g., super exponential in $n$), even writing down approximate values for $\max_j{v(\{j\})}$ and $\min_j{v(\{j\})}$ requires too many bits. To this end, we say that a valuation $v$ is \emph{well bounded} if $\frac {\max_j{v(\{j\})}} {\min_j{v(\{j\})}}\leq n^2$.  We first show that we can focus on sketching well-bounded valuations (a non-algorithmic version of this lemma appeared in \cite{BDFK+12}):
\begin{lemma}\label{lemma-well-bounded-reduction}
Let $A$ be an algorithm that produces an $\alpha$-sketch of well bounded subadditive valuations. Then, there is an algorithm $A'$ that produces an $O(\alpha)$-sketch for any subadditive valuation. Moreover, $A'$ runs $A$ on $n$ sets of items $T_1,\ldots, T_n$, where each item $j$ appears in $O(1)$ sets, in addition to $n$ value queries.
\end{lemma}
Denote by $q$ the number of queries that $A$ makes. Notice that as long as $q\geq n$ queries, the number of queries that $A'$ makes is $O(q)$.

\begin{proof}(proof of Lemma \ref{lemma-well-bounded-reduction})
We construct the $T_i$'s as follows. Without loss of generality, assume that $v(\{1\})\geq v(\{2\}) \geq \ldots \geq v(\{n\})$. For each item $j$, let $N_j=\{j'|\frac {v(\{j\})} {v(\{j'\}}\leq n^2,j\leq j' \}$. Let $T_1=N_1$. Let $j_2$ be the item with smallest index such that $\frac {v(\{1\})} {v(\{j_2\})}\geq \frac n 2$. Let $T_2=N_{j_2}$. Similarly, Let $j_3$ be the item with the smallest index such that $\frac {v(\{j_2\})} {v(\{j_3\})}\geq \frac n 2$, let $T_3=N_{j_3}$ and so on.

Notice that the conditions in the statement of the lemma hold for $T_1,\ldots, T_n$. The sketch that $A'$ produces is obtained by running $A$ on each set $T_i$ separately. All that is left to show is that we can compute any set $S$ to within a factor of $O(\alpha)$. Towards this end, consider some set $S$ and let $j$ be the item with the smallest index in $S$. Let $i$ be the largest index such that $j\in N_i$. Let $S'=S-N_{i}$. Since by assumption we have that we get an $\alpha$-sketch for $S\cap N_{i}$, it suffices to show is that $v(S')\leq \frac {v(S)} {2}$ and this concludes the proof since by subadditivity we then have that $v(S\cap N_{i})\geq \frac {v(S)} {2}$.

Observe that by construction for each $j'\in S'$ we have that $\frac {v(\{j\})} {v(\{j'\})}\geq \frac {n^2 } {\frac n 2}=2n$. Thus, by subadditivity it holds that $v(S')\leq \Sigma_{j'\in S'}v(\{j'\})\leq n\cdot \frac {v(\{j\})} {2n} \leq \frac {v(S)} {2}$, as needed.
\end{proof}

\begin{definition}
Fix a $\beta$-XOS function $v$. Let $S\subseteq N$ be a bundle, and fix some approximate clause $a$ of $S$. For $r\in \mathbb{N}$, the \emph{$r$-projection} of $S$, denoted by $S^{r}$, is defined as
$S^{r}=\{i\in S\,|\; r=2^{l}\leq a\left(i\right)< 2^{l+1}=2\cdot r\}$.

We also define the \emph{core} of $S$, denote $C\left(S\right)$, as the highest valued $r$-projection of $S$, that is: $C\left(S\right)=S^{r'}$, where $a\left(S^{r'}\right)=\max_{r}a\left(S^{r}\right)$.
\end{definition}

\begin{claim} \label{core}
For every bundle $S$ and a well bounded valuation $v$, $v\left(C\left(S\right)\right)\geq a\left(C\left(S\right)\right)\geq \frac{v\left(S\right)}{\beta\cdot 2\log\left(n\right)}$, where $a$ is the additive valuation which is an approximate XOS clause of $S$.
\end{claim}
\begin{proof}
Let $y$ be the number of non-empty projections of $a$.
Recall that $a\left(S\right)\leq \sum_{r}a\left(S^{r}\right)\leq \beta \cdot a\left(S\right)$,
and thus by the pigeonhole principle
we conclude that $\exists r\; s.t.\; a\left(S^{r}\right)\geq\frac{a\left(S\right)}{y}$
and that same $r$ satisfies $v\left(S^{r}\right)\geq a\left(S^{r}\right)\geq\frac{a\left(S\right)}{y}\geq \frac{v\left(S\right)}{\beta \cdot y}$. Since $v$ is well bounded and since the projections were defined on the ranges between powers of $2$, there are at most $2\log\left(n\right)$ non-empty projections. The claim follows by the discussion above.
\end{proof}

%% file: construction.tex
\section{The Construction}

We now describe an algorithm that sketches a well-bounded valuation $v$. Using this algorithm we can obtain sketches for any valuation function by applying Lemma \ref{lemma-well-bounded-reduction}.  Two oracles are used in this algorithm, and we assume that they are given to us as black boxes. The first oracle $CARD(N',k)$ gets as input a subset of the items $N'\subseteq N$ and an integer $k$ and returns a subset $T\subseteq N'$ such that $v(T)\geq \frac {\max_{S\subseteq N',\,|S|\leq k}v\left(S\right)} \alpha$ (i.e., maximization subject to a cardinality constraint).

The second oracle $\beta-XOS(S)$ gets a bundle $S$ and returns a $\beta$-approximate XOS clause of $S$ with respect to the valuation $v$. We will use known implementations of these oracles to obtain our efficient sketches for subadditive and submodular valuations. In the appendix we also show how to obtain sketches for matroid rank function with only $\tilde O(n)$ value queries. For that we provide a new fast and simple implementation of $CARD(N',k)$ for this class. Since $v$ is well bounded we may assume that for all $j$, $1\leq v(\{j\})\leq n^2$.

\subsubsection*{The Sketch Construction Algorithm}

\begin{enumerate}
\item \label{iterations} For every $k=\sqrt n, 2\cdot \sqrt n, 4\cdot \sqrt n, \ldots, n$ and $r=1,2,4,\ldots, n^2$:
\begin{enumerate}
\item Let $\mathcal T_{k,r}=\emptyset$. Let the set of ``heavy'' items be $H_{k,r}=\left\{ i\in N\,|\; v\left(\left\{ i\right\} \right)\geq\frac{k\cdot r}{\sqrt{n}}\right\} $ and the set of ``light'' valued items be $L_r=N-H_{k,r}$. Let $N'=L_r$.
\item \label{loop} While $T=CARD(N',k)$ is such that $v(T)\geq \frac {k\cdot r} {2\alpha}$:
\begin{enumerate}
\item \label{remove_light} Let $a$ be the additive valuation that $\beta-XOS(T)$ returns. Let $T'\subseteq T$ be the set of items $j\in T$ with $a(\{j\})\geq \frac {r} {4\alpha\cdot \beta}$.
\item Add $T'$ to $\mathcal T_{k,r}$. Let $N'=N'-T'$.
\end{enumerate}
\end{enumerate}
\item For every item $j$, return $v(\{j\})$. For every $k,r$ iterated above return $\mathcal T_{k,r}$.
\end{enumerate}

The above algorithm constructs a sketch for a valuation $v$. For a given bundle $S$, we will use the algorithm's output to return its approximate value:

\begin{enumerate}
\item Let $max-singleton= \max_{j\in S}v\left(\{j\}\right)$,
$max-intersection= \max_{k,r,T\in \mathcal T_{k,r}}|T\cap S|\cdot\frac{r}{4\alpha\cdot\beta}$.
\item Return $\max\{max-singleton,\;max-intersection\}$.
\end{enumerate}

\subsubsection*{Analysis}

We now analyze the performance of the algorithm. We start with an helpful claim:

\begin{claim} \label{subsets}
Fix a well bounded subadditive valuation $v$, and consider some bundle $S$, and let $S^{r'}$ be its $r$-projection, where $|S^{r'}|=|S^{r'}\cap L_r|=k'$. Then the algorithm will keep adding subsets to $\mathcal T_{k',r'}$ as long as $|S^{r'}\cap\left(\underset{i}{\bigcup}T_{i}\right)|<\frac{k'}{2}$. Furthermore, at the end of the algorithm we have that $|\mathcal T_{k',r'}|\leq \frac{4\alpha\cdot \beta\cdot \sqrt{n}}{2\beta-1}$.
\end{claim}

\begin{proof}
Consider the disjoint subsets $T_{1},\ldots,T_{l}\in \mathcal T_{k',r'}$. Since we assume $|S^{r'}|=|S^{r'}\cap L_r|=k'$, as long as we did not use at least half of the items in $S^{r'}$ to construct the sets in $\mathcal T_{k',r'}$, then there is some subset $T\subseteq N'$ with $v\left(T\right)\geq\frac{k'\cdot r'}{2}$ and thus the oracle $CARD(\cdot)$ must return some set with value at least $\frac{k'\cdot r'}{2\alpha}$. We conclude that the loop in Step \ref{loop} will not stop before we get that $|S^{r'}\cap\left(\underset{i}{\bigcup}T_{i}\right)|\geq \frac{k'}{2}$.
For the second part of the claim, note that every subset $T$ returned by $CARD(\cdot)$ is valued at least $\frac{k\cdot r}{2\alpha}$. In Step \ref{remove_light} we remove from $T$ the items with value less than $\frac{r}{4\alpha\cdot \beta}$, thus the maximal loss of value in this step will be $\frac{k\cdot r}{4\alpha\cdot \beta}$ and we get that $v\left(T'\right)\geq v\left(T\right)- \frac{k\cdot r}{4\alpha\cdot \beta}= \frac{\left(2\beta-1\right)\cdot k\cdot r}{4\alpha\cdot \beta}$. Observe that the sets consist of only ``light'' items, and since we only deal with subadditive valuations $v\left(\underset{j\in T'} \bigcup{j}\right)\leq \underset{j\in T'} \sum{v\left(\{j\}\right)}$ and we conclude that $|T'|\geq \frac{\frac{\left(2\beta-1\right)\cdot k\cdot r}{4\alpha\cdot \beta}}{\frac{k\cdot r}{\sqrt{n}}}=\frac{\left(2\beta-1\right)\cdot \sqrt{n}}{4\alpha\cdot \beta}$. Since $|N|=n$ we can upper bound the number of the disjoint subsets with $\frac{4\alpha\cdot \beta\cdot \sqrt{n}}{2\beta-1}$.
\end{proof}

The next two lemmas analyze the running time and the approximation ratio of the sketch as a function of the properties of the oracles. We will later consider specific implementations of these oracles and will provide explicit bounds on the performance of the algorithm for each class.

\begin{lemma} \label{queries_number}
For every well bounded subadditive valuation $v$ there is an implementation of the algorithm in which the number of times that the algorithm calls $CARD(\cdot)$ and $\beta-XOS(\cdot)$ is $O(\alpha \cdot \sqrt{n}\cdot \log^2 n)$. In addition, the algorithm makes $O(n)$ value queries.
\end{lemma}
\begin{proof}
Stage \ref{iterations} of the algorithm performs $k\cdot r$ iterations, where $k=O\left(\log\left(n\right)\right)$ and $r=O\left(\log\left(n\right)\right)$. We claim that the loop in \ref{loop} iterates $\frac{4\alpha\cdot \beta\cdot \sqrt{n}}{2\beta-1}$ times: we know from Claim \ref{subsets} that for every $k',r'$, $|\mathcal T_{k',r'}|\leq \frac{4\alpha\cdot \beta\cdot \sqrt{n}}{2\beta-1}$ and this is also the maximal number of iterations of the loop. Overall the number of calls to the oracles $CARD(\cdot)$ and $\beta-XOS(\cdot)$ is indeed $O(\alpha \cdot \sqrt{n}\cdot \log^2 n)$. The additional value queries the algorithm uses are to determine the "heavy" items set and to return the value of every singleton, which overall requires $n$ singleton value queries.
\end{proof}

\begin{lemma}
For every well bounded subadditive valuation $v$, the algorithm returns an $O(\alpha\cdot \beta\cdot  \sqrt{n} \cdot \log n)$-sketch of $v$.
\end{lemma}
\begin{proof}
Fix a bundle $S$. By Lemma \ref{core} the core of $S$ gives us a $\log\left(n\right)$-approximation to $v(S)$, and we claim that the algorithm constructs sketches of each $r$-projection of $S$, and of $C\left(S\right)$ in particular. Let $|C\left(S\right)|=k'$ and $r'$ be the projection index of the core set, then $k'\cdot r' \leq a\left(C\left(S\right)\right)\leq 2\cdot k'\cdot r'$, where $a$ is the additive valuation which is an approximate XOS clause of $S$. Note that if $H_{k',r'}\cap S\neq \emptyset$ then there exists a singleton that gives $O\left(\sqrt{n}\right)$-approximation for $a\left(C\left(S\right)\right)$, and its value will be returned in stage 2 of the algorithm (observe that if $k'<\sqrt{n}$ then by the pigeonhole principle and subadditivity such ``heavy'' item must exist).

If $H_{k',r'}\cap S= \emptyset$, Claim \ref{subsets} gives us that $|C\left(S\right)\cap\left(\underset{i}{\bigcup}T_{i}\right)|\geq \frac{k'}{2}$ and also that $|\mathcal T_{k',r'}|\leq\frac{4\alpha\cdot \beta\cdot \sqrt{n}}{2\beta-1}$. Thus, by the pigeonhole principle, $\exists T_{i}\in\mathcal T_{k',r'}$ such that $|T_{i}\cap C\left(S\right)|\geq \frac{\frac{k'}{2}}{\frac{4\alpha\cdot \beta\cdot \sqrt{n}}{2\beta-1}}=\frac{\left(2\beta-1\right)\cdot k'}{8\alpha\cdot \beta\cdot \sqrt{n}}$. Observe that the value that the algorithm outputs is at least $|T_{i}\cap C\left(S\right)|\cdot\frac{r'}{4\alpha\cdot\beta} \geq \frac{\left(2\beta-1\right)\cdot k'}{8\alpha\cdot \beta\cdot \sqrt{n}}\cdot\frac{r'}{4\alpha\cdot\beta}\geq\frac{\left(2\beta-1\right)\cdot \frac{a\left(C\left(S\right)\right)}{2}}{32\alpha^{2}\cdot \beta^{2}\cdot \sqrt{n}}\geq \frac{ v\left(S\right)}{\frac{128\alpha^{2}\cdot\beta^{3}}{2\beta-1}\cdot \sqrt{n}\cdot \log\left(n\right)}$.
\end{proof}

\subsubsection*{Applications}
By providing specific efficient implementations for the oracles used in the algorithm, we get our main results:

\begin{corollary}\label{cor-submodular}
There exists an algorithm that produces an $\tilde O(\sqrt n)$ sketch of every monotone submodular valuation $v$ using only $O\left(n^\frac{3}{2}\cdot \log^{3}\left(n\right)\right)$ value queries.
\end{corollary}
\begin{proof}
We first prove for a well-bounded valuation $v$. We use a recent result \cite{BV14} that presents an implementation for the oracle $CARD(\cdot)$ for monotone submodular functions with $\alpha=\frac{1}{1-\frac{1}{e}-\epsilon}$ while using $O\left(\frac{n}{e}\cdot \log\left(\frac{n}{e}\right)\right)$ value queries\footnote{Another possibility is using the classic greedy algorithm \cite{NWF78} which gives $\left(1-\frac{1}{e}\right)$-approximation, but its running time is $O\left(n^2\right)$ which results in a slower overall running time.}.
In addition, the implementation of the $\beta-XOS(\cdot)$ oracle for monotone submodular functions is straightforward - given a bundle $S$, we return an additive function $a$ that assigns every item in $S$ with its marginal value, that is: $v\left(i_{1}\right)=v\left(S\right)-v\left(S\backslash \{i_{1}\}\right),\ldots ,v\left(i_{j}\right)=v\left(S\backslash \underset{k<j} \bigcup{i_{k}}\right)-v\left(S\backslash \left(\underset{k<j} \bigcup{i_{k}}\cup\{i_{j}\}\right)\right)$, and get an XOS clause with $\beta=1$ \cite{DNS05}. That requires additional $O\left( n\right)$ value queries. Combining with claim 3.1, we conclude that our algorithm makes $O\left(n^\frac{3}{2}\cdot \log^{3}\left(n\right)\right)$ value queries, and achieves an approximation ratio of $O(\sqrt n \cdot \log n)$. The result for any valuation follows by applying Lemma \ref{lemma-well-bounded-reduction}.
\end{proof}

We do not know whether we can construct sketches for general submodular functions using only $\tilde O(n)$ value queries. However, if $v$ is further known to be a matroid rank function we are able to sketch it using only $\tilde O(n)$ value queries. We postpone this result to the appendix.

\begin{corollary}
There exists an algorithm that produces an $\tilde O(\sqrt n)$ sketch of every subadditive valuation $v$ using $O(n)$ demand and value queries.
\end{corollary}
\begin{proof}
We first prove the result for a well-bounded valuation $v$. For the implementation of the oracle $CARD(\cdot)$ we use a former result \cite{DPS11} that gives an approximation of $\alpha =2$ while using $O\left(\log\left(n\right)\right)$ demand queries. The oracle $\beta-XOS(\cdot)$ can be implemented using a construction that returns an approximate XOS clause with $\beta =O(\log\left(n\right))$ while using $O\left(\log\left(n\right)\right)$ demand queries \cite{D07}. We therefore achieve an approximation ratio of $ O(\sqrt{n}\cdot \log^3 n)$ using $O\left(\sqrt{n}\cdot \log^3\left(n\right)\right)$ demand queries and $O(n)$ value queries. The result for a general valuation $v$ follows by applying Lemma \ref{lemma-well-bounded-reduction}.
\end{proof}

%% file: appendix.tex
\section{Fast Sketches for Matroid Rank Functions}

If $v$ is a matroid rank function, we obtain a $\Omega(\sqrt n)$-sketch of it with only $\tilde O(n)$ value queries.

\begin{proposition}
There exists an algorithm that produces an $\tilde O(\sqrt n)$ sketch of a matroid rank function $v$ using only $O\left(n\cdot \log^2\left(n\right)\right)$ value queries.
\end{proposition}
\begin{proof} (sketch) First, notice that every matroid rank function is well bounded so it is enough to consider this case only. We use the same implementation of the $1$-XOS oracle as in Corollary \ref{cor-submodular} and provide a fast implementation of $CARD(\cdot)$ with an approximation ratio of $1$. Given a set $T$, our implementation finds whether there exists an item $j$ that can be added to $T$ with a marginal value $v(\{j\}+T)-v(T)=1$. Starting from the empty set, $CARD(\cdot)$ will keep adding items to $T$ until there is no item with a marginal of $1$, or until $|T|=k$. Finding an item with a marginal $1$ can be done with $O(\log n )$ queries: consider the marginal value of $N-T$: if it is $0$, by monotonicity there is no item that can be added to $T$ with a marginal of $1$. Else, partition the set of items $N-T$ into two equal sized sets, find the one with the positive marginal value, and recurse until a single item $j$ with a positive marginal value is found. The number of queries we made to find $j$ is $O(\log n)$. The total number of queries that $CARD(\cdot)$ makes is $O(k\cdot \log n)$. Correctness of the algorithm follows from the augmentation property of the matroid: as long as there is an independent set $I$ with $|I|>|T|$ by the augmentation property of matroids there must be an item $j$ such that $T\cup \{j\}$ is an independent set, and thus the marginal value of $j$ given $T$ is $1$.


Observe that Step \ref{iterations} of the sketching algorithm will only perform $O\left(k\right)=O\left(\log\left(n\right)\right)$ iterations, since $r\in\{0,1\}$. In addition, we can remove stage \ref{remove_light} which is now redundant, as the current implementation of $CARD(\cdot)$ never adds items with marginal value of 0.
The overall running time of the loop in \ref{loop} is $O\left(n\cdot \log\left(n\right)\right)$: the running time of each iteration of $CARD(\cdot)$ is $O\left(d_{i}\cdot \log\left(n\right)\right)$, where $d_{i}$ is the rank of the remaining items in $N'$ on the $i^{th}$ iteration of \ref{loop}. Note that $\underset{i} \sum{d_{i}}\leq n$ since the sets are disjoint and clearly $\underset{T_{i}\in \mathcal T_{k}} \sum{|T_{i}|}\leq n$. We conclude that the running time of Step \ref{loop} is indeed $O\left(n\cdot \log\left(n\right)\right)$, and from Lemma \ref{queries_number} we know that the additional number of value queries used by the algorithm is $O\left(n\right)$, thus the algorithm uses a total of $O\left(n\cdot \log^2\left(n\right)\right)$ value queries and achieves an approximation ratio of $O(\sqrt{n}\cdot \log\left(n\right))$.
\end{proof} 